\documentclass[letterpaper, 10 pt, conference]{ieeeconf}  

\IEEEoverridecommandlockouts                              
\usepackage{tikz}
\usepackage{subcaption}
\usepackage{booktabs}
\usepackage{url}
\usepackage{cite}
\usepackage{amsmath,amssymb,amsfonts}
\usepackage{algorithmic}
\usepackage{graphicx}
\usepackage{textcomp}
\usepackage{xcolor}
\def\BibTeX{{\rm B\kern-.05em{\sc i\kern-.025em b}\kern-.08em
    T\kern-.1667em\lower.7ex\hbox{E}\kern-.125emX}}
\newtheorem{prop}{Proposition}
\newtheorem{rem}{Remark }
\newtheorem{ex}{Example}[section]

\usepackage[ruled,vlined,linesnumbered]{algorithm2e}
\usepackage{makecell}

\newcommand{\Real}{\mathbb R}

\newcommand{\norm}[1]{\left\Vert#1\right\Vert}
\newcommand{\ra}{\rightarrow}
\newcommand{\set}[1]{\left\{#1\right\}}

\let\subset\subseteq

\title{\LARGE \bf
Formally Verified Physics-Informed Neural Control Lyapunov Functions}

\author{Jun Liu, Maxwell Fitzsimmons, Ruikun Zhou, and Yiming Meng  
\thanks{This research was supported in part by an NSERC Discover Grant and the Canada Research Chairs program.}
\thanks{Jun Liu, Maxwell Fitzsimmons, and Ruikun Zhou are with the Department of Applied Mathematics, Faculty of Mathematics, University of Waterloo, Waterloo, Ontario N2L 3G1, Canada.  Emails: \texttt{j.liu@uwaterloo.ca, mfitzsimmons@uwaterloo.ca, ruikun.zhou@uwaterloo.ca}
        }%
\thanks{Yiming Meng is with  the
Coordinated Science Laboratory, University of Illinois Urbana-Champaign,
Urbana, IL 61801, USA. Email: 
        \texttt{ymmeng@illinois.edu}}%
}

\begin{document}

\maketitle
\thispagestyle{empty}
\pagestyle{empty}

\begin{abstract}
Control Lyapunov functions are a central tool in the design and analysis of stabilizing controllers for nonlinear systems. Constructing such functions, however, remains a significant challenge. In this paper, we investigate physics-informed learning and formal verification of neural network control Lyapunov functions. These neural networks solve a transformed Hamilton–Jacobi–Bellman equation, augmented by data generated using Pontryagin’s maximum principle. Similar to how Zubov's equation characterizes the domain of attraction for autonomous systems, this equation characterizes the null-controllability set of a controlled system. This principled learning of neural network control Lyapunov functions outperforms alternative approaches, such as sum-of-squares and rational control Lyapunov functions, as demonstrated by numerical examples. As an intermediate step, we also present results on the formal verification of quadratic control Lyapunov functions, which, aided by satisfiability modulo theories solvers, can perform surprisingly well compared to more sophisticated approaches and efficiently produce global certificates of null-controllability.
\end{abstract}

\section{Introduction}

Controlling nonlinear dynamical systems remains one of the most challenging problems in modern control theory. Control Lyapunov functions (CLFs) \cite{artstein1983stabilization,sontag1989universal,sontag2013mathematical,freeman1996control} provide a powerful framework for studying asymptotic stabilization. However, the construction of CLFs is computationally challenging. Traditional approaches, such as sum-of-squares (SOS) \cite{jarvis2003lyapunov} and rational CLFs \cite{doban2015feedback}, can produce valid CLFs, but they remain conservative in capturing the full region of asymptotic (null)  controllability.

Similar to how Zubov's equation \cite{zubov1961methods} captures the domain of attraction for an asymptotically stable autonomous system, it can be extended to systems with inputs \cite{camilli2001generalization,camilli2008control,grune2015zubov} and can provide a characterization of the null-controllability set \cite{camilli2008control}, i.e., the set of initial conditions that can be asymptotically stabilized to an equilibrium point. In this case, Zubov's equation is essentially a transformed Hamilton-Jacobi-Bellman (HJB) equation that maps infinite value to 1, such that the null-controllability set is precisely captured by the 1-sublevel set of the solution to Zubov's equation. Like HJB equations, Zubov's equation is a nonlinear partial differential equation (PDE) that is challenging to solve numerically. In \cite{kang2023data}, a data-driven approach was proposed to solve Zubov's equation using neural networks. In \cite{liu2023towards}, it was shown that adding the PDE loss explicitly in training can lead to a better approximation of the solution to Zubov's equation. They also demonstrated that the physics-informed neural network (PINN) solution to Zubov's equation can be formally verified using satisfiability modulo theories (SMT) solvers to provide rigorous region-of-attraction estimates, outperforming SOS Lyapunov functions \cite{liu2023physics}.

In this paper, we aim to tackle the more challenging case of Zubov's equation for control problems using PINNs. Our approach leverages both trajectory-wise optimization using Pontryagin's maximum principle (PMP) and physics-informed learning to compute CLFs, which are formally verified using SMT solvers. We demonstrate that this approach outperforms traditional methods for computing SOS and rational CLFs. While previous work has combined PMP and neural networks for solving HJB equations \cite{nakamura2021adaptive}, to the best of our knowledge, this paper is the first to explicitly encode the PDE loss for solving the Zubov-HJB equation and provide formally verified CLFs.

Other related work includes the formal synthesis and verification of neural network Lyapunov functions \cite{edwards2024fossil,liu2024lyznet} and the simultaneous learning of formally verified stabilizing controllers \cite{chang2019neural,zhou2022neural,yang2024lyapunov}. Recent research has also addressed learning Lyapunov functions \cite{grune2021computing,gaby2022lyapunov} and CLFs \cite{rego2022learning,grune2023examples} using deep neural networks for potentially high-dimensional systems, relying on training with a Lyapunov inequality loss. However, such computations are typically local in nature (see Example \ref{ex:pendulum} in Section \ref{sec:examples} and examples in \cite{liu2024lyznet}), and the computed Lyapunov functions are not formally verified. The work in \cite{meng2024physics} also used PINNs to solve HJB equations for optimal control, but it relied on an initial stabilizing controller for policy iteration, and the solution is also local.

The remainder of this paper is structured as follows. Section \ref{sec:prelim} introduces the necessary mathematical preliminaries. Section \ref{sec:pinn-clf} presents computation of neural network CLFs using physics-informed learning. Section \ref{sec:verify} details the formal verification process using SMT solvers. Finally, in Section \ref{sec:examples}, we demonstrate the effectiveness of the proposed approach through numerical examples.

\section{Preliminaries}\label{sec:prelim}

\subsection{Nonlinear control system}

Consider a nonlinear control system of the form
\begin{equation}
    \label{eq:sys}
    \dot x  = f(x) + g(x)u,
\end{equation}
where $f:\,\Real^{n}\to\Real^{n}$ and $g:\,\Real^{n}\to \Real^{n\times k}$ are assumed to be locally Lipschitz. It is assumed that $f(0)=0$, i.e., the origin is an equilibrium point in the absence of control input ($u=0$). Given a control signal $u:\,[0,\infty)\to \Real^k$, the corresponding solution to (\ref{eq:sys}) from an initial value $x$ is denoted by $\phi(t;x,u)$, which is defined on a maximal interval of existence $[0,T_{\max}(x,u))$. We refer readers to \cite{sontag1989universal} for basic theory regarding existence and uniqueness of solutions for (\ref{eq:sys}) under admissible control signals.

\subsection{Control Lyapunov function and asymptotic stabilization}

Let $D\subset\Real^n$ be an open set containing the origin. A \textit{control Lyapunov function} $V:\,D\ra \Real$ is a smooth and positive definite function that satisfies 
\begin{equation}
    \label{eq:clf}
\inf_{u\in\Real^k} \left\{\nabla V\cdot (f(x) + g(x)u)\right\}<0
\end{equation}
for all $x\in D\setminus\set{0}$. If $D=\Real^n$ and $V$ is radially unbounded (i.e., $V(x)\ra\infty$ as $\norm{x}\ra\infty$), it is called a global control Lyapunov function. A well-known result in nonlinear control is that the existence of a control Lyapunov function implies the existence of a stabilizing feedback law of the form $u = k(x)$, which satisfies \( k(0) = 0 \), is continuous at the origin, and smooth everywhere else \cite{artstein1983stabilization}. Sontag's formula \cite{sontag1989universal} provides an explicit construction of such a stabilizing feedback law, given as 
\begin{equation}
    \label{eq:sontag}
    k = -\frac{a+\sqrt{a^2+\norm{b}^4}}{\norm{b}^2}b^\top,
\end{equation}
when $\norm{b}\neq 0$, and $k=0$, when $\norm{b}=0$. Here  
\begin{align*}
a(x)&=L_f V(x):=\nabla V(x)\cdot f(x)\\
b(x)&=L_g V(x):=\nabla V(x)\cdot g(x).
\end{align*}
The CLF condition (\ref{eq:clf}) can be equivalently stated as 
\begin{equation}
    \label{eq:clf2}
b(x) = 0 \Longrightarrow a(x)<0
\end{equation}
for all $x\in \setminus\set{0}$, because when $b(x) \neq 0$, (\ref{eq:clf}) can be easily satisfied. In fact, when $b(x) \neq 0$, applying Sontag's formula (\ref{eq:sontag}) gives
\begin{equation}
    \nabla V\cdot (f + g k) = a - (a + \sqrt{a^2+b^4}) = -\sqrt{a^2+b^4} < 0. 
\end{equation}
The subtle part in the analysis of \cite{sontag1989universal} lies in proving that the feedback law $u=k(x)$ given by (\ref{eq:sontag}) is continuous at the origin (under the small control property \cite{sontag1989universal}) and smooth elsewhere, even when $b(x)$ could be zero. 

\subsection{Connections with optimal control}

By inverse optimality \cite{freeman1996inverse}, every CLF is a meaningful value function for an optimal control problem. Consider an infinite-horizon cost
\begin{equation}
    \label{eq:cost}
    J(x,u)=\int_0^\infty \left[ q(x) + u^\top R(x) u \right] dt,
\end{equation}
where $q$ and $R$ are positive definite functions. The infinite-horizon optimal control problem seeks to find a control signal that minimizes this cost.  By a standard dynamic programming argument, one can reduce the problem to the celebrated Hamilton-Jacobi-Bellman (HJB) equation \cite{bardi1997optimal}: 
\begin{equation}
    \label{eq:hjb}
    \sup_{u\in\Real^k} \left\{-\nabla V (f + g u) - u^\top R u - q \right\}= 0. 
\end{equation}
Since $R$ is positive definite, we can rewrite (\ref{eq:hjb}) as
\begin{equation}
    \label{eq:simplified_hjb}
    -\nabla V (f + g k) - k^\top R k - q= 0, 
\end{equation}
where 
\begin{equation}
    \label{eq:k_V}
    k = -\frac12 R^{-1} g^\top \nabla V^\top.
\end{equation}
If $V$ is an optimal value function that satisfies (\ref{eq:simplified_hjb}), then $k$ in (\ref{eq:k}) is an optimal controller. If $V$ is positive definite and continuously differentiable, then by (\ref{eq:simplified_hjb}) and that $q$ is positive definite, $V$ is clearly a CLF. 

Conversely, given a CLF \(V\), while there can be different asymptotically stabilizing control laws, a pointwise min-norm formulation always produces an inverse optimal control law \cite{freeman1996inverse}. An interesting connection between Sontag's formula and the pointwise min-norm formulation is demonstrated in \cite{freeman1996control}, which shows how Sontag's formula can alternatively be viewed from an optimal control perspective.

\subsection{Null controllable set and CLF via Zubov's equation}

Let $\mathcal{U}=L^\infty([0,\infty),\Real^k)$ denote the set of admissible control signals. Similar to the domain of attraction for an autonomous system, the domain of null-controllability for (\ref{eq:sys}) is defined by 
\begin{equation*}
\mathcal{D}:= \set{x\in\Real^n:\, \exists u\in \mathcal{U} \text{ s.t. } \phi(t;x,u)\ra 0 
\text{ as } t\ra \infty}.
\end{equation*}
Zubov's method \cite{zubov1961methods} is well-known for characterizing the domain of attraction for autonomous systems. An extension in \cite{camilli2008control} shows that Zubov's method can also characterize the domain of null-controllability for control systems. The following transformed HJB equation is considered in \cite{camilli2008control}:
\begin{equation}
    \label{eq:zubov-hjb0}
        \sup_{u\in\Real^k} \left\{-\nabla v (f + g u) -  (1 - v)(u^\top R u + q) \right\}= 0, 
\end{equation}
which is essentially a transformation of the HJB equation (\ref{eq:hjb}) with \(v = 1 - \exp(-V)\). This transformation effectively changes the infinite value to 1 and, similar to Zubov's equation defines the domain of attraction \cite{zubov1961methods}, the value function \(v\) now characterizes the null-controllability set as \(\mathcal{D} = \set{x \in \Real^n : v(x) < 1}\). We refer readers to \cite{camilli2008control} for a detailed theoretical investigation of Zubov's equation, in a more general form than (\ref{eq:zubov-hjb0}), and its connections with CLF.

\subsection{Pontryagin's maximum principle} \label{sec:pmp}

While the HJB approach emphasizes closed-loop feedback control, Pontryagin's maximum principle (PMP) provides an alternative perspective, focusing on trajectory-wise optimization. Define the Hamiltonian \( H \) as
\[
H(x, u, \lambda) = q(x) + u^\top R(x) u + \lambda^\top \left[ f(x) + g(x) u \right],
\]
where \( \lambda \in \mathbb{R}^n \) is the adjoint (co-state) vector. According to PMP, the  optimal control \( u^*\) minimizes the Hamiltonian with respect to \( u \), leading to:
\[
\frac{\partial H}{\partial u} = 0 \quad \implies \quad 2 R(x) u^* + g(x)^\top \lambda = 0.
\]
Solving for \( u^* \), we obtain the optimal control law:
\begin{equation}
u^*(t) = -\frac{1}{2} R(x(t))^{-1} g(x(t))^\top \lambda(t).
\label{eq:ustar}
\end{equation}
The adjoint equation is derived as
\begin{align*}
\dot{\lambda}(t) &= -\frac{\partial H}{\partial x} \\
& = -\left[ \frac{\partial q}{\partial x} + \left( \frac{\partial f}{\partial x} \right)^\top \lambda + \left( \frac{\partial g}{\partial x} u^* \right)^\top \lambda + u^{*\top} \frac{\partial R}{\partial x} u^* \right].
\end{align*}
The transversality condition requires the adjoint vector to vanish at infinity, i.e., 
$
\lim_{t \to \infty} \lambda(t) = 0.
$

Let \( V(t) \) represent the accumulated cost-to-go from time \( t \) along the optimal trajectory. It satisfies 
\[
\dot{V}(t) = -\left[ q(x(t)) + u^{*\top}(t) R(x(t)) u^*(t) \right],
\]
Since $V$ represents the cost-to-go and we are interested in characterizing asymptotic stabilization, it is expected that $\lim_{t \to \infty} V(t) = 0.$

The necessary conditions from PMP allow us to formulate optimization algorithms to approximate both the optimal trajectories and the values of \( V \) along them. By selecting a sufficiently large terminal time \( T > 0 \), the boundary conditions for \( \lambda \) and \( V \) are approximately \( V(T) \approx 0 \) and \( \lambda(T) \approx 0 \). At \( t = 0 \), we impose the initial condition \( x(0) = x_0 \). This results in the following two-point boundary value problem (TPBVP):
\begin{equation}
\label{eq:tpbvp}
\begin{aligned}
    \dot{x} &= f(x) + g(x) u^*, \\
    \dot{\lambda} &= - \frac{\partial q}{\partial x} + \left( \frac{\partial f}{\partial x} \right)^\top \lambda + \left( \frac{\partial g}{\partial x} u^* \right)^\top \lambda + u^{*\top} \frac{\partial R}{\partial x} u^*, \\
    \dot{V} &= -\left[ q(x) + u^{*\top} R(x) u^* \right],
\end{aligned}
\end{equation}
where $u^*$ is given by (\ref{eq:ustar}), subject to the boundary conditions:
\begin{equation}
    \label{eq:boundary}
    x(0) = x_0, \quad V(T) = 0, \quad \lambda(T) = 0.
\end{equation}
By approximately solving this TPBVP using off-the-shelf solvers, we can obtain approximate values of $V$ along optimal trajectories. 

\section{Physics-informed computation of maximal neural network control Lyapunov functions} \label{sec:pinn-clf}

\subsection{Zubov-HJB equation}

Given the steady-state HJB equation (\ref{eq:hjb}), similar to \cite{liu2023physics}, we propose a transformation \(\beta: [0,\infty) \rightarrow \Real\) with the property that \(\beta(0) = 0\), \(\beta\) is strictly increasing, and \(\beta(s) \rightarrow 1\) as \(s \rightarrow \infty\). It is easy to check that any \(\beta\) satisfying the differential equation
\begin{equation}\label{eq:beta}
    \dot{\beta} = (1 - \beta)\psi(\beta)
\end{equation}
with \(\psi(s) > 0\) for all \(s \geq 0\), satisfies the above properties. Two straightforward choices of \(\psi\) are \(\psi(s) = \alpha\) and \(\psi(s) = \alpha(1 + s)\) for some $\alpha>0$, which correspond to \(\beta(s) = 1 - \exp(-\alpha s)\) and \(\beta(s) = \tanh(\alpha s)\), respectively. The former is known as the Kruzkov transform. Both are commonly used in the study of Zubov's equation \cite{camilli2001generalization,kang2023data,liu2023physics}.

One potential issue that impedes the computation of the optimal value function \(V\) is that the unique solution to (\ref{eq:simplified_hjb}) is only well-defined for \(x \in \mathcal{D}\) and approaches infinity as \(x\) approaches the boundary of \(\mathcal{D}\) \cite{camilli2008control}. Hence, we aim to compute a new function \(W\) such that \(W(x) = \beta(V(x))\) for all \(x \in \mathcal{D}\) and \(W(x) = 1\) elsewhere. In other words, the null-controllability set will be characterized by \(\set{x \in \Real^n : W(x) < 1}\). The fact that \(W\) is bounded by 1 can potentially make computation easier. By (\ref{eq:beta}), we have
\begin{equation}
    \label{eq:dW}
    \nabla W = (1-W)\psi(W) \nabla V.
\end{equation}
and 
$W$ should satisfy the PDE
\begin{equation}
    \label{eq:zubov_hjb}
    \begin{aligned}
   & -\nabla W (f (1-W)\psi(W) + g \hat k) - {\hat k}^\top R \hat k \\
    &\qquad  - q (1-W)^2\psi^2(W)= 0, 
    \end{aligned}
\end{equation}
where $\hat k = - \frac12 R^{-1} g^\top \nabla W^\top.$ We emphasize that writing the equation in the form of (\ref{eq:zubov_hjb}) to avoid any term involving \(1 - W\) in the denominator is important for numerical stability when solving (\ref{eq:zubov_hjb}) with neural networks. We refer to equation (\ref{eq:zubov_hjb}) as the Zubov-HJB equation. Note that, by (\ref{eq:dW}), the optimal controller $k$ from (\ref{eq:k_V}) is now given as 
\begin{equation}
    \label{eq:k}
    k 
    = -\frac{1}{2(1-W)\psi(W)} R^{-1} g^\top \nabla W^\top.
\end{equation}

\subsection{Physics-informed neural solution to Zubov-HJB equation}\label{sec:pinn-hjb}

We can write (\ref{eq:zubov_hjb}) as a general first-order PDE 
\begin{equation}
    \label{eq:1st_order_PDE}
    F(x,W,DW) = 0,\quad x\in \Omega,
\end{equation}
subject to the boundary condition $W = h$ on $\partial\Omega$.

Let $W_N(x; \theta)$ represent a multi-layer feedforward neural network parameterized by $\theta$. We aim to train $W_N$ to solve (\ref{eq:1st_order_PDE}) by minimizing the loss function 
\begin{align}
    \text{Loss}(\theta) &= \frac{1}{N_c}\sum_{i=1}^{N_c} F(x_i,W_N(x_i;\theta),DW_N(x_i;\theta))^2 \notag \\
    &\quad + \lambda_b \frac{1}{N_b} \sum_{i=1}^{N_b}(W_N(y_i;\theta)-h(y_i))^2,  \notag  \\
    &\quad + \lambda_d \frac{1}{N_d} \sum_{i=1}^{N_d}(W_N(z_i;\theta)-\hat W(z_i))^2. 
    \label{eq:lossV}
\end{align}
Here, $\lambda_b > 0$ and $\lambda_d > 0$ are weighting parameters. In (\ref{eq:lossV}), the points $\set{x_i}_{i=1}^{N_c} \subset \Omega$ are referred to as (interior) collocation points. At these points, we evaluate $W_N(x;\beta)$ and its derivative to compute the mean-square residual error $\frac{1}{N_c} \sum_{i=1}^{N_c} F(x_i, W_N(x_i;\beta), DW_N(x_i;\beta))^2$ for (\ref{eq:zubov_hjb}). Although we use the mean-square error here, any error measure that captures the magnitude of the residual error could be applied. The set $\set{y_i}_{i=1}^{N_b} \subset \partial \Omega$ consists of boundary points where the mean-square boundary error $\frac{1}{N_b} \sum_{i=1}^{N_b} (W_N(y_i; \beta) - h(y_i))^2$ is computed. Additionally, a set of data points $\set{(z_i, \hat{W}(z_i))}_{i=1}^{N_d}$ can be used, where $\set{z_i}_{i=1}^{N_d} \subset \Omega$ and $\set{\hat W(z_i)}$ are approximate ground truth values for $W$ at $\set{z_i}_{i=1}^{N_d}$. In the next subsection, we will discuss in detail how to use PMP, presented in Algorithm \ref{alg:pmp}, to generate data points for solving (\ref{eq:zubov_hjb}).

The idea behind this approach is that by minimizing the residual and boundary losses, the underlying physics of the problem is incorporated, allowing the neural network $W_N$ to closely approximate the (unique) solution of the PDE. The parameters $\theta$ are usually trained using gradient descent methods. Physics-informed neural networks for solving PDEs have recently gained considerable interest among scientific communities \cite{raissi2019physics,lagaris1998artificial}. 

\subsection{Data generation using Pontryagin's maximum principle}

Algorithm \ref{alg:pmp} summarizes how to generate data for the transformed value function $W$ by solving a TPBVP from Pontryagin's maximum principle outlined in Section \ref{sec:pmp}. The procedure is fairly standard, and there are many off-the-shelf TPBVP solvers. Furthermore, computations from different initial conditions can be easily done in parallel. While TPBVP solvers only offer local convergence guarantees, as they use Newton-type algorithms, by setting the tolerance accuracy very high, one can expect such solvers to be accurate on successfully solved cases. In general, it would be numerically challenging, if not impossible, to accurately solve the TPBVP from any initial condition within the null-controllability set. This makes solving the Zubov-HJB equations purely based on PMP data challenging. We demonstrate that with limited data and the loss function (\ref{eq:lossV}) encoding the PDE (\ref{eq:zubov_hjb}), we can effectively solve the Zubov-HJB equation using neural networks. We shall illustrate this in Section \ref{sec:examples} with numerical examples.

\begin{algorithm}[ht]
\caption{Generate PMP Data}\label{alg:pmp}
\KwIn{System dynamics $f$, $g$, cost $q$, $R$, number of samples \( n_{\text{samples}} \), terminal time \( T \), number of time steps \( N \), tolerance \( \text{tol} \)}
\KwOut{Dataset of initial states and value functions}
\BlankLine

\For{\( i = 1 \) to \( n_{\text{samples}} \)}{
    Sample initial state \( x_0 \)\;
    
    Solve TPBVP (\ref{eq:tpbvp}) for \( x(t), \lambda(t), V(t) \) with \( N \) time steps and tolerance \( \text{tol} \), subject to boundary conditions (\ref{eq:boundary})\;
        
    Compute transformed value function from \( W(x_0)=\beta(V(0)) \)\;
    
    Store \( x_0 \) and transformed value function $W(x_0)$ \;
}

\Return Collected data\;
\end{algorithm}

\section{Formal verification of control Lyapunov functions}\label{sec:verify}

\subsection{Formal verification of quadratic control Lyapunov functions}

Consider the control system (\ref{eq:sys}). Let 
$
A = \frac{\partial f}{\partial x}\big\vert_{x=0}
$
and
$B = g(0). 
$
In this subsection, we assume that the pair $(A,B)$ is stabilizable.
Suppose that $Q$ and $R$ are symmetric matrices such that $(A,Q)$ is detectable and $R$ is positive definite. By linear system theory, one can solve the algebraic Riccati equation
\begin{equation}
    \label{eq:are}
    PA+A^\top P - PBR^{-1}B^\top P + Q= 0
\end{equation}
to obtain a unique positive definite solution $P$. Furthermore, a feedback stabilizing controller law is given by $u=Kx$ with 
\begin{equation}
    \label{eq:linear_control}
    K = - R^{-1}B^\top P.
\end{equation}

We can easily prove following results on the existence of a local quadratic control Lyapunov function for (\ref{eq:sys}). 

\begin{prop}
    Suppose that $f$ is continuously differentiable and $g$ is continuous. Assume that $(A,B)$ is stabilizable and $Q$ is positive definite. There exists some $c>0$ such that 
    $V_P(x)=x^\top Px$, where $P$ is the solution to (\ref{eq:are}), is a control Lyapunov function for (\ref{eq:sys}) on an open set containing $\set{x\in\Real^n:\, V_P(x)\le c}$. 
\end{prop}

\begin{proof}
    Consider the closed-loop system under the linear control law $u=Kx$ with $K$ from (\ref{eq:linear_control}) and write it as 
    \begin{align*}
    \dot x &= f(x) + g(x) K x \\
    &= (A+BK)x + (f(x)-Ax) + (g(x)-B)Kx.
    \end{align*}
    We can compute that 
    \begin{align*}
        & \nabla V_P \cdot (f(x) + g(x) K x)  = 2x^\top P (A+BKx) \\
        & \qquad + 2x^\top (f(x)-Ax) + 2x^\top P(g(x)-B)Kx \\
        & = x^\top (PA +A^\top P - PBR^{-1} B^\top P +Q) x \\
        &\qquad - x^\top PBR^{-1} B^\top P x - x^\top Qx + o(\norm{x}^2)\\
        & \le - x^\top Qx + o(\norm{x}^2),
    \end{align*}
    where \(o(\norm{x}^2)\) represents higher-order terms than \(\norm{x}^2\). Because \(Q\) is positive definite, we conclude that the CLF condition (\ref{eq:clf}) holds on $\set{x\in\Real^n:\, V_P(x)\le c}$ for \(c > 0\) sufficiently small.
\end{proof}

While the proof shows the existence of a local quadratic control Lyapunov function, we can leverage a formal verifier such as an SMT solver to compute the largest value $c>0$ on which $V_P$ remains a valid CLF. 

We formulate two CLF conditions to be verified:
\begin{itemize}
\item[(1)] \textbf{global CLF condition}:
\begin{equation}
\nabla V_P \cdot g = 0 \land x \in \Real^n\setminus \set{0} \implies  \nabla V_P \cdot f < 0.
\label{eq:global_qclf}
\end{equation}

\item[(2)] \textbf{local CLF condition}: 
\begin{equation}
\nabla V_P \cdot g = 0 \land x \neq 0 \land V_P(x)\le c \implies  \nabla V_P \cdot f < 0.
\label{eq:local_qclf}
\end{equation}

\end{itemize}

Algorithm \ref{alg:qclf} summarizes the verification procedure. It either globally verifies (\ref{eq:global_qclf}) or finds the largest \(c \in (0, c_{\max}]\), for some \(c_{\max} > 0\), such that the local CLF condition (\ref{eq:local_qclf}) holds.

\begin{algorithm}[ht]
\SetAlgoLined
\KwIn{System dynamics $f$, $g$, candidate  CLF $V_P(x)=x^\top Px$, bisection upper bound $c_{\max}$}
\KwOut{Largest \( c_P\in (0,c_{\max}] \) for which the CLF condition (\ref{eq:local_qclf}) holds}

\If{CLF condition (\ref{eq:global_qclf}) holds globally}{
    Return success\;
}
\Else{
    Use bisection to find largest $c_P\in (0,c_{\max}]$ such that the local CLF condition (\ref{eq:local_qclf}) holds\;
}

\Return{$c_P$}\;

\caption{Quadratic CLF Verifier}
\label{alg:qclf}
\end{algorithm}

\begin{rem}
While conceptually simple, we state Algorithm \ref{alg:qclf} separately to emphasize that verification of quadratic CLFs can be highly effective in practice and often outperforms more sophisticated computational approaches for constructing CLFs (see Examples \ref{ex:ex1} and \ref{ex:pendulum}). This is due to the existence of complete algorithms for verifying inequalities (\ref{eq:global_qclf}) and (\ref{eq:local_qclf}). If both $f$ and $g$ are polynomials, the SMT solver Z3 \cite{de2008z3} provides complete and sound procedures for verifying (\ref{eq:global_qclf}) and (\ref{eq:local_qclf}). If $f$ and $g$ are non-polynomial, we can leverage the $\delta$-complete SMT solver dReal \cite{gao2013dreal}. While dReal cannot verify the global CLF condition (\ref{eq:global_qclf}), it is effective in verifying (\ref{eq:local_qclf}) on bounded domains. The solver dReal being $\delta$-complete means it either formally proves (\ref{eq:local_qclf}) or finds a counterexample to a slightly weakened version of the negation of (\ref{eq:local_qclf}) up to a prescribed precision value. 
\end{rem}

\begin{rem}
We implemented Algorithm \ref{alg:qclf} in LyZNet \cite{liu2024lyznet} with support for both Z3 and dReal. Since dReal is only $\delta$-complete, counterexamples near the origin are inevitable. To address this issue, we use the same approach as in \cite{liu2023physics} (also implemented in LyZNet \cite{liu2024lyznet}) to verify that $V_P$ is a valid Lyapunov function on $\{x \in \mathbb{R}^n : V_P(x) \leq c\}$ for the closed-loop system under the linear control $u = Kx$, where $K$ is given by (\ref{eq:linear_control}). Denote the largest verifiable $c$ by $c_P^1$. This implies that the CLF condition (\ref{eq:local_qclf}) holds on $\{x \in \mathbb{R}^n : V_P(x) \leq c_P^1\}$. We can then verify
\[
\nabla V_P \cdot g = 0 \land x \neq 0 \land c_P^1 \leq V_P(x) \leq c \implies \nabla V_P \cdot f < 0
\]
to determine a larger $c_P$ as in Algorithm \ref{alg:qclf}. The above condition can be effectively verified by dReal because it already excludes the region $\{x \in \mathbb{R}^n : V_P(x) \leq c_P^1\}$ containing the origin in its interior. Note that $c_P$ is usually larger than $c_P^1$, because the CLF condition (\ref{eq:clf}) allows the use of any control, not just the linear controller. 
\end{rem}

To conclude this subsection, we note that if a global quadratic CLF is formally verified, then the domain of null-controllability is $\mathcal{D}=\Real^n$. If a local quadratic CLF is formally verified using Algorithm \ref{alg:qclf}, we have \(\{x \in \mathbb{R}^n : V_P(x) \leq c_P\} \subset \mathcal{D}\). In both cases, we can effectively construct a controller, e.g., using Sontag's formula (\ref{eq:sontag}), to achieve asymptotic stabilization on the formally verified region of null-controllability. 

\subsection{Formal verification of neural control Lyapunov functions}

While a quadratic CLF can be simple and effective, there are situations where a more complex CLF can provide a better estimate of the null-controllable set. In fact, the main purpose of this paper is to demonstrate that by solving Zubov's PDE in a principled fashion as detailed in Section \ref{sec:pinn-clf}, we can obtain a near-maximal estimate of the null-controllability set using a verified region of null-controllability in terms of sublevel sets of a neural CLF.

Let $W_N$ denote a candidate neural network CLF. We propose the following conditions to be formally verified:
\begin{itemize}
\item[(1)] \textbf{set containment:} 
\begin{equation}
\nabla W_N \le c_1 \implies V_P \le c_P,
\label{eq:neural_clf1}
\end{equation}
where $V_P$ and $c_P$ are formally verified by Algorithm \ref{alg:qclf}.

\item[(2)] \textbf{CLF condition:}
\begin{align}
\nabla W_N \cdot g = 0 & \land x \neq 0 \land c_1 \le W_N(x)\le c_2 \notag\\
& \implies  \nabla W_N \cdot f < 0.
\label{eq:neural_clf2}
\end{align}
\end{itemize}

Similar to Algorithm \ref{alg:qclf}, we can use bisection to determine the maximal $c_2$ such that the CLF condition (\ref{eq:neural_clf2}) can be formally verified. By doing so, we formally prove that 
\(\{x : W_N(x) \leq c_2\} \subset \mathcal{D}\). A controller can be effectively constructed using Sontag's formula with $V_N$ when $c_1 \leq W_N(x) \leq c_2$ and using $V_P$ when $V_P(x) \leq c_P$. Because of the CLF condition (\ref{eq:neural_clf2}), the closed-loop system, from any initial state in \(\{x : W_N(x) \leq c_2\}\) and under this controller, will reach \( \{x : W_N(x) \leq c_1\} \), which is contained in \( \{x : V_P(x) \leq c_P\} \) by the set containment condition (\ref{eq:neural_clf1}), and be asymptotically stabilized to the origin from \(\{x : V_P(x) \leq c_P\}\) since $V_P$ is a verified CLF on this set.

The verification procedure for a neural CLF is summarized in Algorithm \ref{alg:neural_clf}.

\begin{algorithm}[ht]
\SetAlgoLined
\KwIn{System dynamics \( f, g \), candidate neural CLF \( W_N \), verified quadratic CLF \( V_P \) on \(\{x : V_P(x) \leq c_P\}\), bisection upper bound \( c_{\max} \)}

\KwOut{Largest \( c_1 \in (0,c_{\max}] \) for which the set containment condition (\ref{eq:neural_clf1}) holds and largest \( c_2 \in (c_1, c_{\max}] \) for which the neural CLF condition (\ref{eq:neural_clf2}) holds}

Use bisection to find the largest \( c_1 \in (0, c_{\max}] \) such that the set containment condition (\ref{eq:neural_clf1}) holds\;

Use bisection to find the largest \( c_2 \in (c_1, c_{\max}] \) such that the neural CLF condition (\ref{eq:neural_clf2}) holds\;

\Return{$c_1, c_2$}\;

\caption{Neural CLF Verifier}
\label{alg:neural_clf}
\end{algorithm}

\subsection{Formal verification of near-optimal controller}

Solving the transformed Zubov-HJB equation allows us to compute a near-maximal estimate of the null-controllability set and extract a formally verified near-optimal controller. The following result demonstrates that a near-optimal stabilizing controller can be obtained by solving the Zubov-HJB equation (\ref{eq:zubov_hjb}).

\begin{prop}
Suppose that $W$ solving (\ref{eq:zubov_hjb}) is positive definite and satisfies $\mathcal{D}=\set{x\in\Real^n:\, W(x)<1}$. Then the closed-loop system for (\ref{eq:sys}) with $u=k(x)$ given by (\ref{eq:k}) is asymptotically stable with $\mathcal{D}$ as the domain of attraction. 
\end{prop}

\begin{proof}
    We prove that $W$ is a Lyapunov function for the closed-loop system on $\mathcal{D}$. In fact, by re-arranging (\ref{eq:zubov_hjb}), we have
    \begin{align}
        \nabla W (f+gk) &= -\frac{\hat{k}^\top R\hat k}{(1-W)\psi(W)} - Q(1-W)\psi(W) \notag\\
        &\le -Q(1-W)\psi(W),\label{eq:lyap_W}
    \end{align}
    which is negative definite on $\mathcal{D}$.
\end{proof}

\begin{rem}
As noted in Section \ref{sec:pinn-hjb}, the premise of solving the Zubov-HJB equation (\ref{eq:zubov_hjb}) using neural networks is that we can arbitrarily approximate the true solution $W$ to (\ref{eq:zubov_hjb}) with a neural network $W_N$. To this end, theoretical guarantees for solving the Zubov-HJB equation (\ref{eq:zubov_hjb}) with PINNs are of significant interest. We leave this as future work. Nonetheless, we provide a heuristic analysis of why the closed-loop system can be made asymptotically stable despite the use of neural network solutions. Let $W_N$ be a neural network solution to (\ref{eq:zubov_hjb}). Define
\begin{equation}
\label{eq:neural_control}
k_N = -\frac{1}{2(1-W)\psi(W)} R^{-1} g^\top \nabla W^\top.
\end{equation}
Since we know $k(0)=0$ and $\nabla k(0) = K$, it is expected that $k_N(0)\approx 0$, $\nabla k_N(0) \approx K$, and $k_N\approx k$ on any compact domain on which $W_N\approx W$ and $\nabla W_N \approx \nabla W$. Now we can simply subtract $k_N(0)$ from $k_N$, still denoted by $k_N$ for simplicity, and have $k_N(0)=0$, $\nabla k_N(0)\approx K$, and $k_N\approx k$ elsewhere. As a result, the origin remains an equilibrium of the closed-loop system under controller $k_N$, and the quadratic Lyapunov function $V_P$ remains valid for local stability analysis. Furthermore, away from the origin, (\ref{eq:lyap_W}) ensures that $W_N$ remains a valid Lyapunov function for the closed-loop system under controller $k_N$.
\end{rem}

Consider the closed-loop system 
\begin{equation}\label{eq:closed-loop-neural}
\dot x = f(x) + g(x)k_N,
\end{equation}
where $K_N$ is given by (\ref{eq:neural_control}). 
Since this is an autonomous system, the verification of asymptotic stability is similar to that in \cite{liu2023physics} (implemented in the tool LyZNet \cite{liu2024lyznet}). The difference is that \( W_N \) is obtained by solving the Zubov-HJB equation instead of the Zubov equation.










\section{Numerical examples}\label{sec:examples}

In this section, we present numerical examples to illustrate the main algorithms presented in this paper. All numerical examples were carried out using LyZNet \cite{liu2024lyznet}, a Python tool designed for learning and verifying neural network Lyapunov functions and regions of attraction. The code for these examples is available in the ``examples/pinn-clf'' directory of the repository: \url{https://git.uwaterloo.ca/hybrid-systems-lab/lyznet/}. All training and verification tasks were performed on a 2020 MacBook Pro with a 2 GHz QuadCore Intel Core i5, without any GPU. 

\begin{ex}\label{ex:ex1}
    Consider the following Van der Pol equation with input:
    \begin{align*}
        \dot x_1 = x_2,\quad \dot x_2 = -x_1 + x_2(1 - x_1^2) + u,
    \end{align*}
    and the mass-spring system \cite{jarvis2003lyapunov}: 
        \begin{align*}
        \dot x_1 & = x_2,\quad \\
        \dot x_2 & = -(x_1 + \frac{1}{10}x_1^3) + (x_3 - x_1) - \frac{1}{10}(x_4 - x_2) + u,\\
        \dot x_3 & = x_4,\quad \dot x_4 = -(x_3 - x_1) + \frac{1}{10}(x_4 - x_2).
    \end{align*}
    For both examples, we can verify that the quadratic value function produced by solving the Riccati equation (\ref{eq:are}) with $Q=R=I$ (identity matrix) is a global CLF using Algorithm \ref{alg:qclf}, and the SMT solver Z3 \cite{de2008z3} can verify the global CLF condition within approximately 20 milliseconds. We further tested scalability of verifying the quadratic CLF by increasing the number of masses in the mass-spring system, where the first spring, which connects the first mass to a fixed wall, stiffens and generates a nonlinear force based on displacement. The rest of the chain is connected with a linear spring and a negative damper between masses. For $N=6$, which defines a 12-dimensional system, we can successfully verify a global quadratic CLF within a second. 
\end{ex}

\begin{ex}\label{ex:pendulum}
    Consider the inverted pendulum
    $$
    \dot x_1 = x_2,\quad \dot x_2 = \frac{g_c}{\ell}\sin(x_1) - \frac{b}{m\ell^2}x_2 + \frac{u}{m\ell^2},
    $$
   where the parameters $g_c = 9.81$, $b = 0.1$, $\ell = 0.5$, and $m = 0.15$ are taken from the example in \cite{rego2022learning}. The authors of \cite{rego2022learning} trained a neural network with 1,545 trainable parameters as a local CLF. Using Algorithm \ref{alg:qclf}, we can show that the quadratic CLF is verifiably global. Although Z3 does not support trigonometric functions, we implemented rigorous polynomial bounds \cite{bagul2022generalized} for common trigonometric functions in LyZNet \cite{liu2024lyznet} and were able to verify that the quadratic value function obtained by solving the Riccati equation (\ref{eq:are}) with $Q = R = I$ is a global CLF in 10 milliseconds.
\end{ex}

\begin{ex}\label{ex:vdp}
Consider the reversed Van der Pol equation with control from \cite{doban2015feedback}:
\begin{equation}
\label{eq:reverse_vdp}
\dot x_1 = -x_2, \quad \dot x_2 = x_1 + (1+u)(x_1^2 - 1)x_2.
\end{equation}
It can be shown that this system does not admit a global quadratic CLF. In \cite{doban2015feedback}, the authors used rational CLFs to investigate the feedback stabilization of this system. We use LyZNet to compute a quadratic CLF \( V_P \) using (\ref{eq:are}) with \( Q = R = I \) and verify it using Algorithm \ref{alg:qclf}. A comparison of the largest verified level set \( \set{x : V_P(x) \leq c_P} \) using Algorithm \ref{alg:qclf}, a rational Lyapunov function from \cite{doban2015feedback}, and a sixth-degree sum-of-squares (SOS) CLF using the procedure from \cite{jarvis2003lyapunov} is shown in Figure \ref{fig:compare_qclf_data}. We observe that the quadratic CLF verified by an SMT solver is comparable to the rational CLF and outperforms the SOS CLF.

We next demonstrate that by solving the Zubov-HJB (\ref{eq:zubov_hjb}) using neural networks and formally verifying them as CLFs, we can significantly expand the estimate of the null-controllability set. Following Algorithm \ref{alg:pmp}, we first solve the TPBVP (\ref{eq:tpbvp}) from 3,000 random initial conditions within the domain $[-4,4]$. We set $T=200$, $N=2,000$, and $\text{tol}=1e-5$. The cost function is set as $q(x)=x^TQx$ and $Q=R=I$. We use \texttt{solve\_bvp} from the Python library \texttt{SciPy} to solve it. We successfully solved it from 1,254 initial conditions, and the values of the transformed value function $W=\beta(V)$, with $\beta(s)=\tanh(\alpha s)$ and $\alpha=0.1$, are used as data points in training the neural network solution $W_N$ to (\ref{eq:zubov_hjb}) using the loss (\ref{eq:lossV}). The data points are visualized in right panel of Figure \ref{fig:compare_qclf_data}. We did not use any additional boundary conditions. We set the neural network to a standard feedforward neural network with hyperbolic tangent activation functions. It has two hidden layers, and each layer uses 30 hidden units. We train the neural network with $300,000$ interior collocation points, randomly sampled from a larger domain $[-8,8]$, for 20 epochs with a batch size of 32. We then formally verify the neural CLF $W_N$ using Algorithm \ref{alg:neural_clf} and the SMT solver dReal \cite{gao2013dreal}. The formally verified level set, along with the learned neural CLF, is depicted in Figure \ref{fig:neural_clf}. It can be seen that the neural CLF significantly outperforms the verified quadratic CLF, which is already on par with alternative approaches, as shown in left panel of Figure \ref{fig:compare_qclf_data}. 

A natural question to ask is whether the physics-informed PDE loss (\ref{eq:lossV}) brings any benefits, as PMP data potentially already captures the optimal value function sufficiently well. To draw a comparison, we train a neural network entirely based on the PMP data points and formally verify it using Algorithm \ref{alg:neural_clf}. We then demonstrate that adding the PDE loss, even when the PMP data is generated on a smaller domain $[-4,4]$, can enhance the training of a neural network solution to the PDE (\ref{eq:zubov_hjb}) on a much larger domain $[-18,18]$. In contrast, fitting a neural network with data alone fails to extrapolate to the larger domain. The comparison is depicted in Figure \ref{fig:data_zubov_compare}. We note that effectively generating data beyond the domain $[-4,4]$ seems numerically challenging. 

Finally, we demonstrate the benefit of solving the Zubov-HJB in achieving a near-optimal controller while maintaining a near-maximal formally verified estimate of the null-controllability set. We form the closed-loop system (\ref{eq:closed-loop-neural}) using the neural controller (\ref{eq:neural_control}) derived from the neural CLF \(W_N\) and formally verify its stabilizing properties using LyZNet. We verified that the level set \(\{x: W_N(x) \leq c\}\) for \(c = 0.8317\) is a region of attraction for the closed-loop system, only slightly smaller than the largest \(c = 0.8334\), where the CLF condition (\ref{eq:neural_clf2}) holds for \(W_N\). We compare the accumulated cost of the Sontag controller (\ref{eq:sontag}) and the HJB controller (\ref{eq:neural_control}) in Figure \ref{fig:cost}, with the HJB controller showing reduced accumulated cost as expected.
\end{ex}

\begin{figure}[ht]
    \centering
    \includegraphics[height=0.33\linewidth]{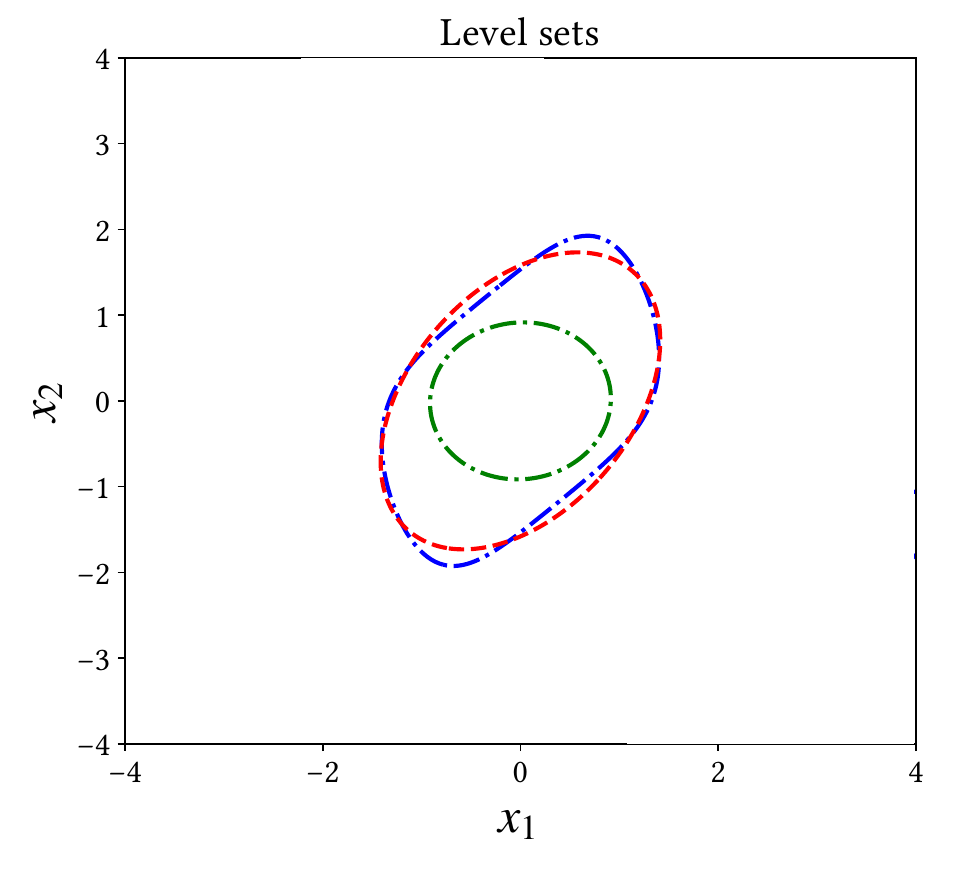}
    \includegraphics[height=0.36\linewidth]{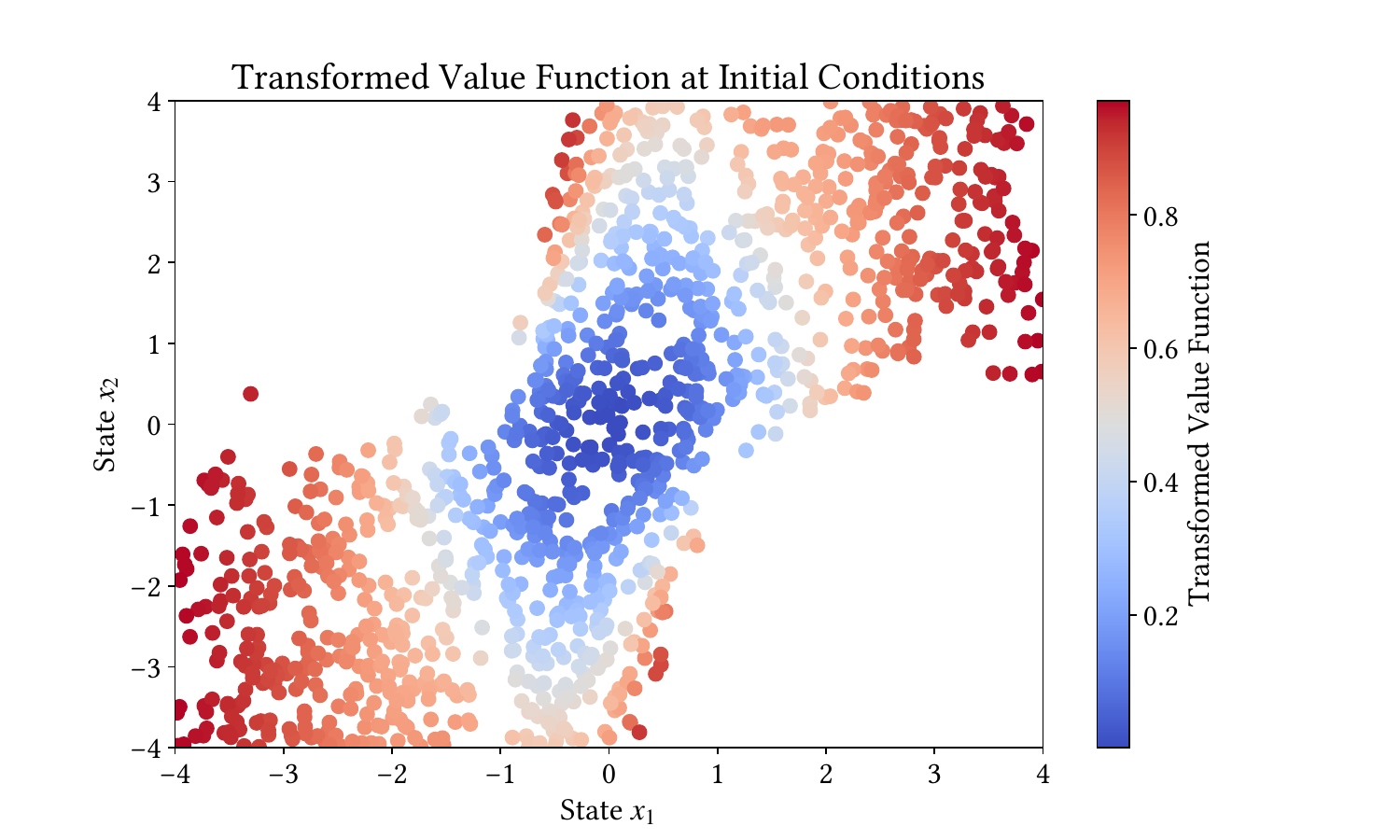}
    \caption{(\textbf{Left}) Comparison of the formally verified null-controllable set provided by a quadratic CLF using Algorithm \ref{alg:qclf} (dashed red), a rational CLF from \cite{doban2015feedback} (dashed blue), and a sum-of-squares (SOS) CLF with a sixth-degree polynomial (dashed green). It is observed that the quadratic CLF verified by an SMT solver is already comparable to the rational CLF and outperforms the SOS CLF. (\textbf{Right}) Data obtained by solving the TBPVP (\ref{eq:tpbvp}) for (\ref{eq:reverse_vdp}) in Example \ref{ex:vdp} on the domain $[-4,4]$. Solving is successful for 1,524 out of 3,000 randomly sampled initial conditions.}
    \label{fig:compare_qclf_data}
\end{figure}


\begin{figure}[ht]
    \centering
    \includegraphics[width=\linewidth]{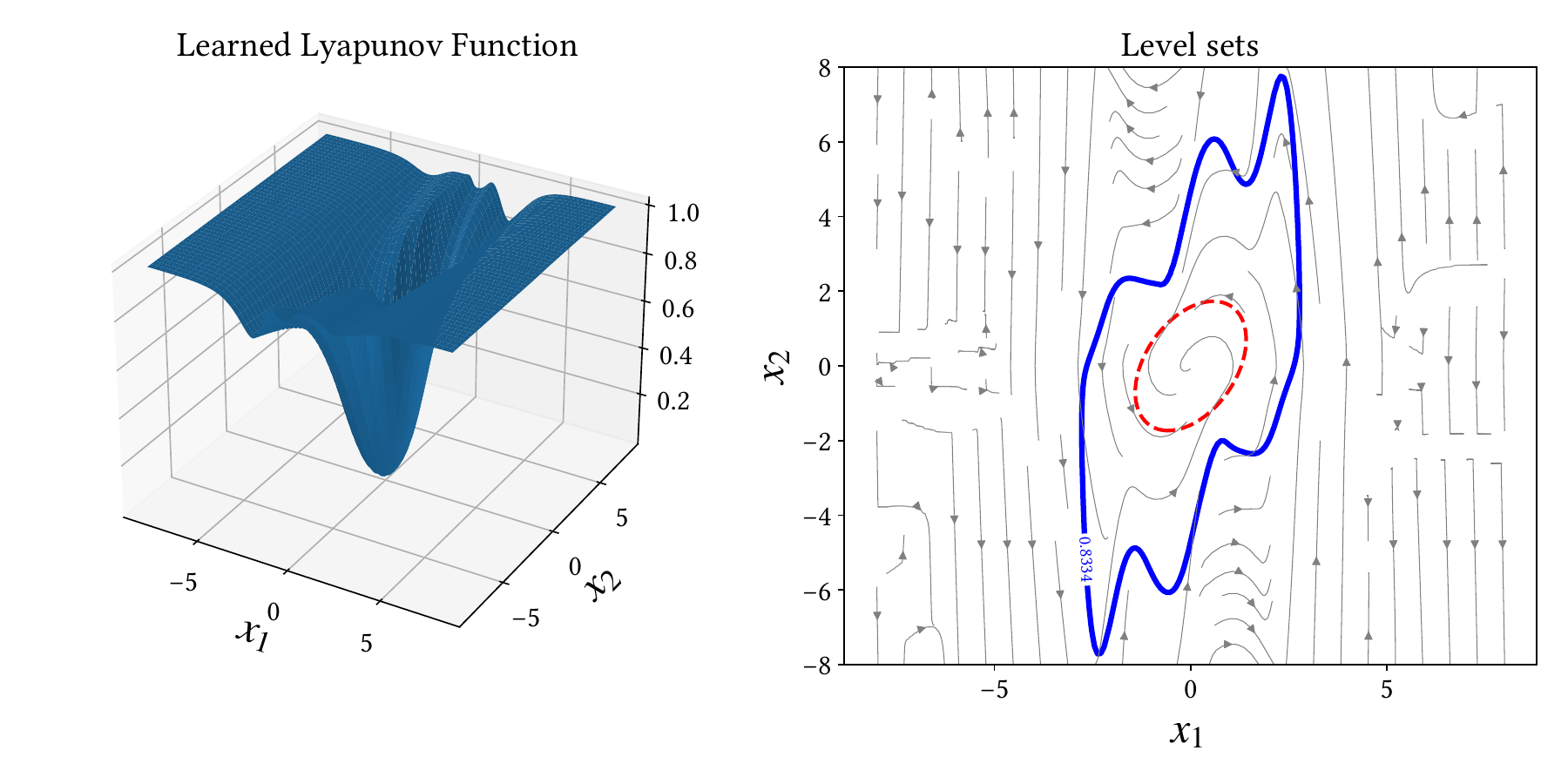}
    \caption{A formally verified neural CLF, obtained by solving the Zubov-HJB (\ref{eq:zubov_hjb}) with a neural network of two hidden layers, each with 30 neurons, can significantly outperform alternative approaches. The solid blue line represents the formally verified level set of the neural CLF, while the dashed red line represents that of a quadratic CLF (see left of Figure \ref{fig:compare_qclf_data} for comparisons). The phase portrait is for the closed-loop system under Sontag's controller (\ref{eq:sontag}) using the learned neural CLF.}
    \label{fig:neural_clf}
\end{figure}

\begin{figure}[ht]
    \centering
    \includegraphics[width=\linewidth]{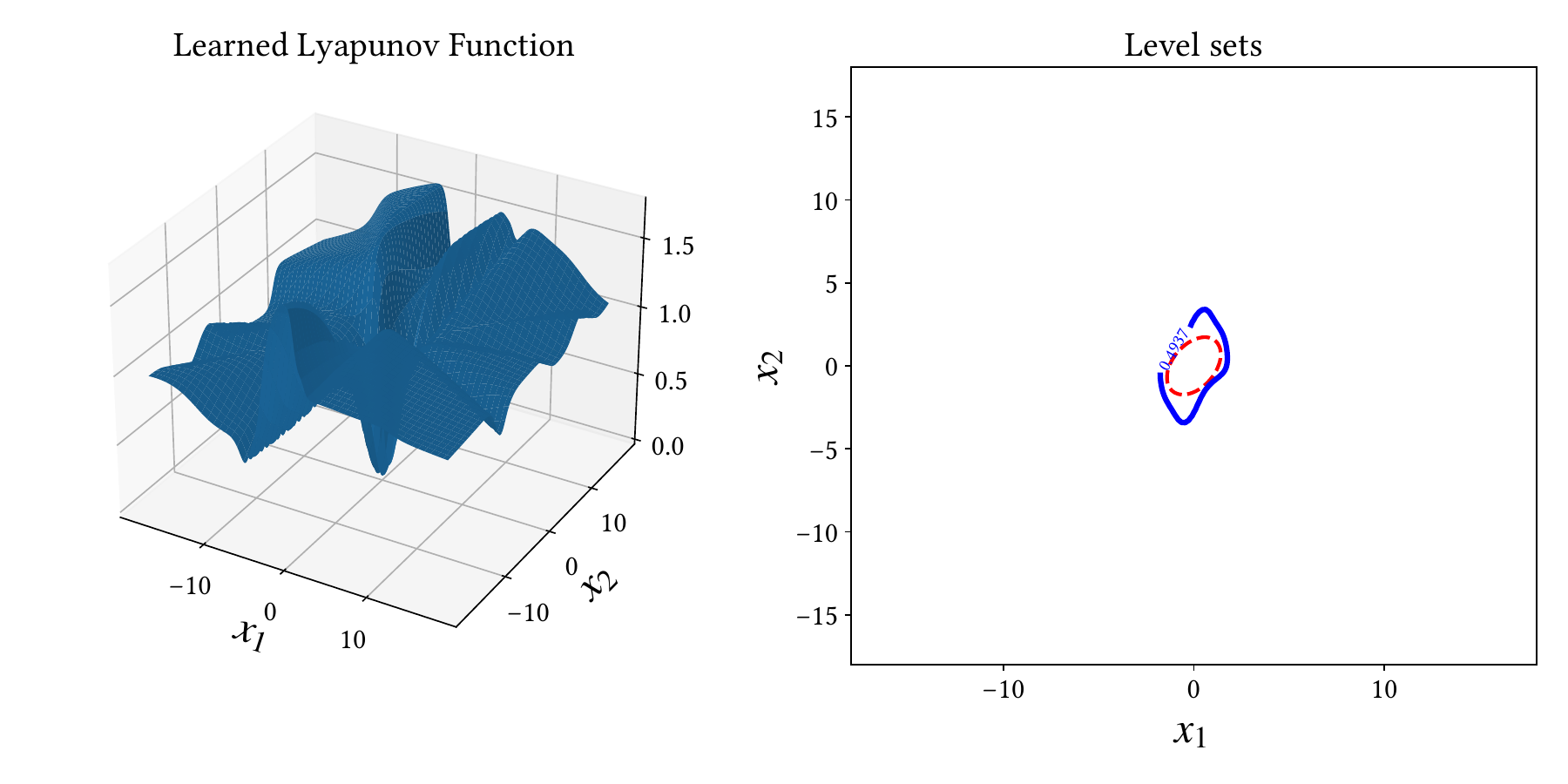}    \\
    \includegraphics[width=\linewidth]{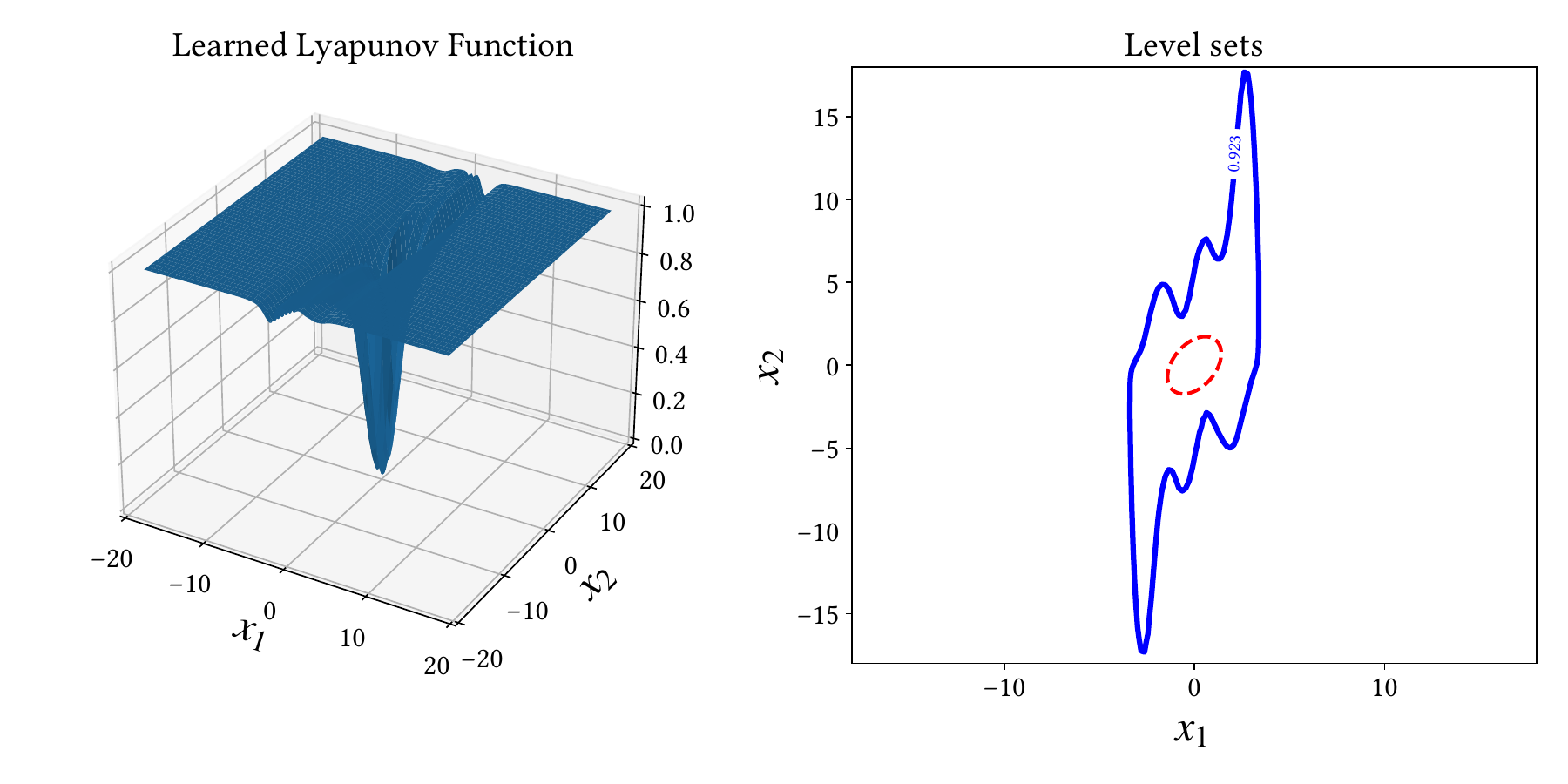}
    \caption{A comparison of the neural network CLF obtained with PMP data alone and with the physics-informed PDE loss (\ref{eq:lossV}). Clearly, the PDE loss significantly enhances extrapolation on larger domains beyond where the data were taken. \textbf{Top:} verified neural CLF trained with PMP data alone. \textbf{Bottom:} verified neural CLF trained with the physics-informed PDE loss (\ref{eq:lossV}).}
    \label{fig:data_zubov_compare}
\end{figure}

\begin{figure}[ht]
    \centering
    \includegraphics[width=0.45\linewidth]{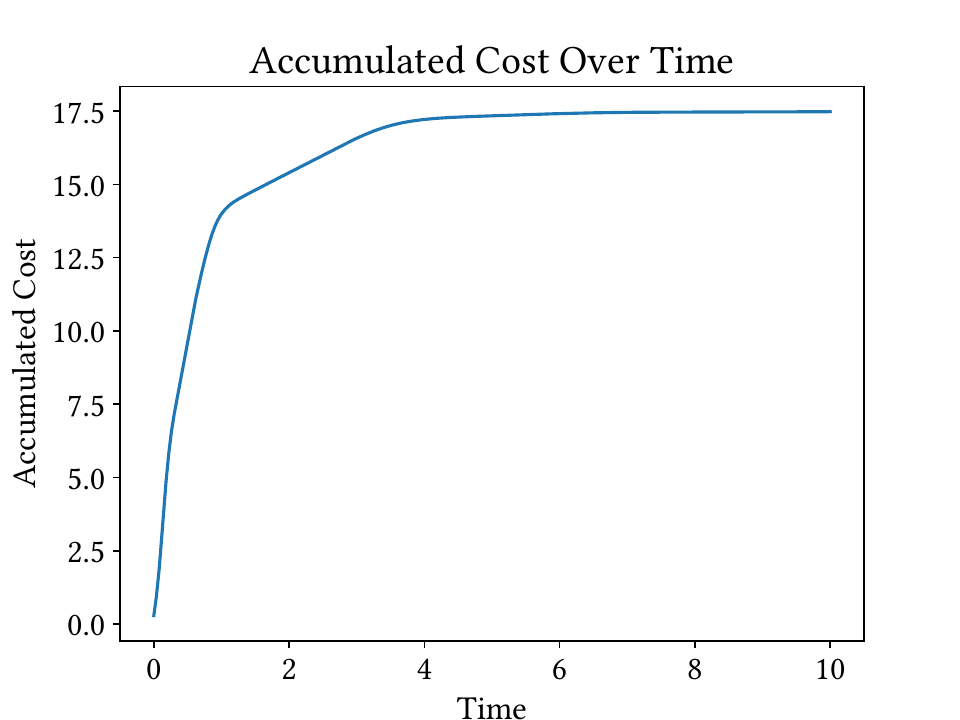}    
    \includegraphics[width=0.45\linewidth]{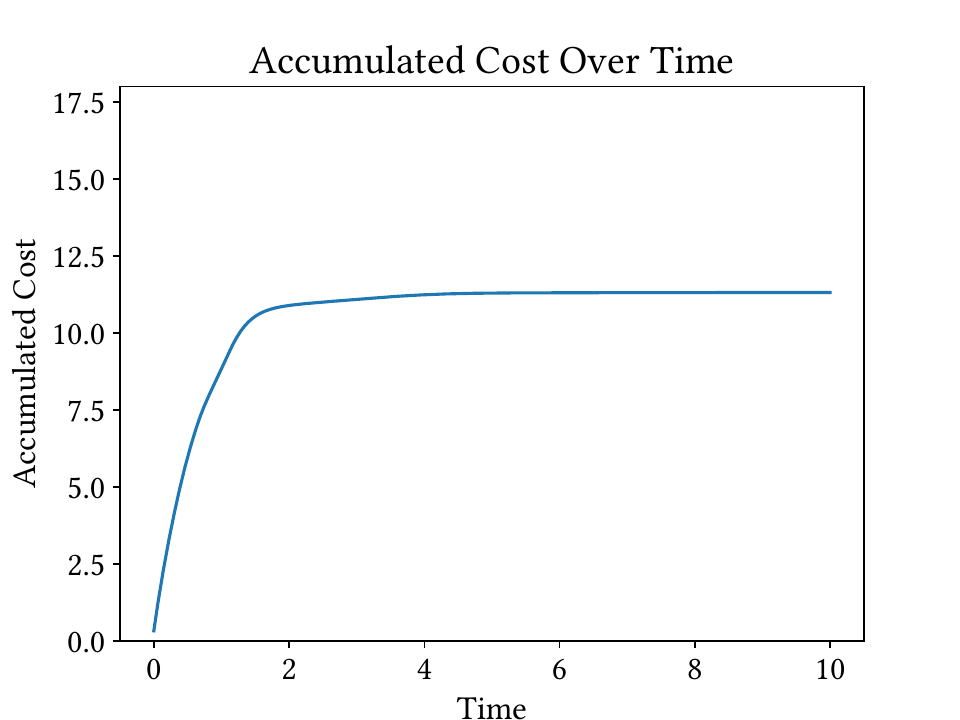}
    \caption{A comparison of the accumulated cost of the Sontag controller (\ref{eq:sontag}) and the neural HJB controller (\ref{eq:hjb}) from the same neural CLF/value function $W_N$. The initial condition is set to be $(-8/3,-8/3)$. It is clear that the HJB controller has a lower accumulated cost, as it is expected to be near optimal by solving the Zubov-HJB equation (\ref{eq:zubov_hjb}).}
    \label{fig:cost}
\end{figure}

\section{Conclusions}

We presented a method for computing and formally verifying neural network-based control Lyapunov functions using physics-informed learning. By solving a transformed Hamilton-Jacobi-Bellman equation, termed the Zubov-HJB equation, the neural CLFs were shown to significantly outperform traditional approaches such as sum-of-squares and rational CLFs in estimating the null-controllability set for nonlinear systems. Additionally, we demonstrated that the neural CLFs can be formally verified using SMT solvers. By approximately solving the optimal value functions, our approach also allows for the derivation of near-optimal controllers with formal guarantees of stability.

Future work will focus on extending this framework to handle bounded control inputs, state constraints, and more complex systems using compositional verification techniques or vector CLFs. We will also investigate robust CLF formulations to address uncertainties and external disturbances in practical control applications.

\bibliographystyle{plain}
\bibliography{acc25}

\end{document}